\documentclass[print]{IEEEtran}
\usepackage{generic}
\usepackage{cite}
\usepackage{amsmath,amssymb,amsfonts}
\usepackage{graphicx}
\usepackage{textcomp}
\def\BibTeX{{\rm B\kern-.05em{\sc i\kern-.025em b}\kern-.08em
    T\kern-.1667em\lower.7ex\hbox{E}\kern-.125emX}}
\let\labelindent\relax
\usepackage{enumitem}

\newtheorem{theorem}{Theorem}[section]

\newtheorem{lemma}[theorem]{Lemma}
\newtheorem{assumption}{Assumption}
\newtheorem{definition}{Definition}
\newtheorem{proposition}[theorem]{Proposition}
\newtheorem{corollary}[theorem]{Corollary}

\newcommand{\Pbar}{\Bar{P}}
\newcommand{\Id}{\mathbb{I}}
\newcommand{\gradJ}{\nabla_{s_i}J_i}
\newcommand{\Jtild}{\widetilde{J}}
\newcommand{\Ftild}{\widetilde{F}}
\newcommand{\Gtild}{\widetilde{G}}
\newcommand{\Ptild}{\widetilde{P}}
\newcommand{\Gbar}{\bar{G}}
\newcommand{\sbar}{\Bar{s}}
\newcommand{\zbar}{\Bar{z}}
\newcommand{\Rset}{\mathbb{R}}
\newcommand{\Eset}{\mathbb{E}}
\newcommand{\Ncal}{\mathcal{N}}
\newcommand{\Scal}{\mathcal{S}}
\newcommand{\Fcal}{\mathcal{F}}
\newcommand{\zi}[2]{z_i(#1\vert#2)}
\newcommand{\smax}{s_{\rm max}}
    
\begin{document}

\title{Learning in Time-Varying Monotone Network Games with Dynamic Populations}

\author{
Feras Al Taha, Kiran Rokade, Francesca Parise
\thanks{This work was partially supported  by the Natural Sciences and Engineering Research Council of Canada, the National Science Foundation grant ECCS-2340289 and  the Cornell Engineering Sprout Awards program.}
\thanks{
The authors are with the School of Electrical and Computer Engineering, Cornell University, Ithaca, NY, 14853, USA (e-mail:  {\tt\small\{foa6,kvr36,fp264\}@cornell.edu)}}
}

\maketitle

\begin{abstract}
In this paper, we present a framework for multi-agent learning in a nonstationary dynamic network environment.
More specifically, we examine projected gradient play in smooth monotone repeated network games in which the agents'  participation and connectivity vary over time.
We model this changing system with a stochastic network which takes a new independent realization at each repetition.
We show that the strategy profile learned by the agents through projected gradient dynamics over the sequence of network realizations converges to a Nash equilibrium of the game in which players minimize their expected cost, almost surely and in the mean-square sense. We then show that the learned strategy profile is an almost Nash equilibrium of the game played by the agents at each stage of the repeated game with high probability. Using these two results, we derive non-asymptotic bounds on the regret incurred by the agents. 
\end{abstract}


\section{Introduction}

Many modern socio-technical systems involve large numbers of  autonomous agents influencing each other over a network.
Instances of such  systems can be found in financial markets, online advertising, supply chain management, urban transportation, or wireless networks. Due to computational limitations or incomplete information, strategic agents in these settings will often be unable to determine in advance the most effective strategy to follow and will instead employ learning techniques to guide their actions. While performance of such learning schemes has been extensively studied in static environments, the aforementioned applications typically involve  a high volume of interactions across large populations leading to  time-variability.
This paper aims at providing a framework for studying multi-agent learning in such time-varying network settings, with a focus on gradient dynamics in smooth monotone network games subject to changing populations and changing network connectivity.

\subsection{Main contributions}

We consider a setting in which, at each repetition of the learning dynamics, players participation is probabilistic and the network of interactions is sampled from a given random graph model. By focusing on projected gradient dynamics 
for smooth and strongly monotone games, we derive almost sure and mean-square convergence guarantees of the players' strategies to a Nash equilibrium of the expected network game (i.e., a network game in which players minimize the expectation of their cost functions).
Additionally, we show that, at each repetition, the learned strategy profile is an $\epsilon$-Nash equilibrium for the stage game, with high probability, and derive bounds on the time-averaged regret that hold almost surely and in expectation.

A preliminary version of this work  appeared in \cite{altaha2023gradient}. 
Therein, almost sure convergence of gradient play to an almost Nash equilibrium of the expected game was shown for the specific class of linear-quadratic network games.
The current paper extends those results to the larger class of smooth and strongly monotone network games, derives bounds on the convergence rate (both in expectation and almost surely) and provides non-asymptotic regret guarantees.

\subsection{Related works}

Conditions for convergence to a Nash equilibrium  in \textit{static network games} have been extensively explored. 
Examples of learning dynamics studied in such literature include passivity-based schemes \cite{gadjov2018passivity,de2019distributed}, asynchronous distributed algorithms \cite{koshal2016distributed,zhu2022asynchronous}, best response schemes \cite{parise2019variational,parise2020distributed} and proximal algorithms \cite{paccagnan2016aggregative,salehisadaghiani2018distributed,gadjov2020single,grammatico2017proximal}, among others.

Prior work has also investigated learning in \textit{time-varying games}.
Under the assumption that the variation of Nash equilibria or the variation in payoff matrices diminishes over time, several papers propose algorithms that are guaranteed to converge and exhibit sublinear dynamic regret \cite{zhang2022no,duvocelle_online_learning}.
When instead the equilibria vary at a nondecreasing rate, the goal becomes to establish regret bounds that grow at the same rate of variation \cite{duvocelle_online_learning,yan2023fast,anagnostides2024convergence,lu2020online}.
For zero-sum games with payoff matrices that vary periodically, \cite{fiez2021online} shows that a broad class of learning algorithms  converges to a time-invariant set.
More broadly, for time-varying variational inequalities (VI), \cite{hadiji2024tracking} examines both 
 settings in which the VI solution varies sublinearly and settings in which the operator defining the VI changes periodically. Our work complements these results by focusing on the specific class of network games and by modeling explicitly how players' participation probabilities, network size and network variability affect the learning process. 
 
The dynamic nature of fluctuating populations is investigated also in the literature on \textit{open multi-agent systems}.
Therein, one common population model (used e.g. in consensus dynamics \cite{hendrickx2017open} and optimal resource allocation problems \cite{de2024random}) consists of a  population with fixed size in which however agents can drop out of the system to be immediately replaced by new agents.
This model was used in game theoretic applications as well, to characterize the efficiency of outcomes in games such as public good games, congestion games, markets, or first-price auctions \cite{lykouris2016learning,bergemann2010dynamic,cavallo2009efficient,cavaliere2012prosperity}.
Similar models of dynamic populations have also been studied for truthful mechanism design \cite{dolan1978incentive,parkes2003mdp, cavallo2009efficient}.
A more general population model used in consensus dynamics allows the population size to vary with agents entering at any time or dropping out without  replacement \cite{de2020modelling,franceschelli2020stability}.
In contrast to these papers, our dynamic population model assumes a fixed set of agents whose participation at each stage is probabilistic, thus allowing the number of participating agents to vary and agents to rejoin after leaving.
A similar probabilistic model has been used to represent random day-to-day agent participation in a multi-agent reinforcement learning setting in \cite{fiscko2023model} but with a fixed underlying network of interactions.
In another similar open multi-agent system framework
\cite{oliva2023sum}, agents can leave and rejoin at any time, creating arbitrary new links in the network when joining.
This differs from the stochastic network model we propose in which a new independent network realization is generated at each stage.

\subsection{Paper organization}

The rest of the paper is organized as follows.
Section \ref{sec:recap} introduces the framework of network games and summarizes known results on gradient dynamics for repeated games over a static network.
Section \ref{sec:dyn_net} presents our model for time-varying networks with dynamic populations, and properties of gradient play in such setting.
Convergence of these gradient dynamics is demonstrated in Section \ref{sec:conv_regr} along with regret guarantees.
Section \ref{sec:concl} concludes the paper.

\subsection{Notation}
Let $[y]_i$ denote the $i$-th component of the vector $y\in\Rset^n$ and $X_{ij}$ denote the $ij$-th entry of the matrix $X\in\Rset^{m\times n}$.
The symbol $\mathbf{1}_n$ represents the $n$-dimensional vector of all ones and the symbol $\Id_n$ represents the $n\times n$ identity matrix.
We use $\|\cdot\|_2$ to denote the Euclidean norm.
We also use $\|\cdot\|_F$ and $\|\cdot\|_2$ to denote the Frobenius norm and the matrix norm induced by the Euclidean norm.
Given $n$ vectors $y_1,\dots,y_n$, $[y_i]_{i\in\{1,\dots,n\}}$ represents the vector formed by stacking $y_1,\dots,y_n$ vertically.
Given $n$ square matrices $X_1, \dots, X_n$, $\textrm{diag}(X_1,\dots,X_n)$ represents the matrix formed by stacking $X_1, \dots, X_n$ block-diagonally.
The Kronecker product of two matrices is denoted by the symbol $\otimes$ and the projection of a point $s$ onto the set $\mathcal{X}$ is denoted by $\Pi_\mathcal{X}[s]$.
Given an operator $F: \Rset^{Nn} \rightarrow \Rset^{Nn}$, and a set $\Scal \subseteq \Rset^{Nn}$, we say that $\sbar \in \Scal$ is a solution of the variational inequality VI$(F,\Scal)$ if $(s-\sbar)^\top F(\sbar) \ge 0 \ \forall s \in \Scal$.

\section{Preliminaries} \label{sec:recap}

In this section, we review the network game model that will be used to represent the interactions of strategic agents.
We  then recap known results on the outcome of learning dynamics in  static environments.

\subsection{One shot game} \label{sec:stage_game}

We consider a game in which agents indexed by $i\in\Ncal:=\{1,\dots,N\}$ interact over a static network represented by an adjacency matrix $\Gtild \in [0,1]^{N\times N}$, with no self-loops, i.e., $\Gtild_{ii} = 0$ for all $i \in \Ncal$.
Each agent $i\in\Ncal$ plays a strategy $s_i \in \Scal_i\subseteq \Rset^n$ and incurs a cost
\begin{align}
    J_i(s_i, \zi{s}{\Gtild})
\end{align}
where $J_i:\Rset^{n}\times\Rset^{n} \to \Rset$ 
is the cost function, and 
\begin{align}
    \zi{s}{\Gtild} := \frac1N \sum_{j=1}^N \Gtild_{ij} s_j
\end{align}
denotes the local aggregate sensed by agent $i$ over the network $\Gtild$, given a strategy profile $s:= [s_i]_{i\in\Ncal}$. 
We require that the cost functions and the strategy sets satisfy the following classic assumption. 
\begin{assumption}[Differentiability, convexity, and compactness] \label{a:strat_set}
    For each $i\in\Ncal$, 
    \begin{enumerate}
        \item the cost function $J_i(s_i, z_i)$ is continuously differentiable in $(s_i,z_i)$, and convex in $s_i$ for all $z_i$. 
        \item the strategy set $\Scal_i$ is nonempty, convex, and compact. Moreover, $\exists \, \smax >0$ (independent of $N$) 
        such that for all $i \in \Ncal$ and all $s_i \in \Scal_i$, $\|s_i\|_2 \leq \smax$.
        \item   $\exists \, \bar{J}_1 > 0, \exists \, \bar{J}_2 > 0$ (independent of $N$) such that $|J_i(s_i,z_i)| \leq \bar{J}_1$ and $\|\nabla_{s_i} J_i(s_i, z_i)\|_2 \leq \bar{J}_2$ for all admissible $s_i, z_i$.  
    \end{enumerate}
\end{assumption}
Note that the assumption that $\smax$, $\bar{J}_1$, and $\bar{J}_2$ are independent from $N$ is typically met in practice because the size of each agent's strategy set does not scale with $N$ and their cost typically only depends on their own action and their local aggregate.

A common solution concept for studying the outcome of a game is the Nash equilibrium, in which agents have no incentive to unilaterally change their strategies.
\begin{definition}[$\epsilon$-Nash equilibrium]
Let $\epsilon \geq 0$. 
An \textit{$\epsilon$-Nash equilibrium} is any strategy profile $\sbar\in\Rset^{nN}$ for which, for all $i\in\Ncal$, it holds that $\sbar_i\in\Scal_i$ and 
\begin{align} \label{eq:eps_NE}
    J_i(\sbar_i, \zi{\sbar}{\Gtild}) \le J_i(s_i, \zi{\sbar}{\Gtild}) + \epsilon \qquad \forall s_i \in \Scal_i.
\end{align}
If $\sbar$ satisfies \eqref{eq:eps_NE} with $\epsilon=0$, then $\sbar$ is a \textit{Nash equilibrium}.
\end{definition}

By defining the \textit{game Jacobian} operator
\begin{align}
    F(s, \Gtild) &:= \Big[\gradJ(s_i,\zi{s}{\Gtild})\Big]_{i\in\Ncal} \, ,
\end{align}
one can characterize the Nash equilibria of a convex game and relate them to the solution of a variational inequality defined by the operator $F$ above and the solution set  $\mathcal{S} := \prod_{i \in \Ncal} \Scal_i$ (see e.g. \cite{facchinei2003finite,scutari2010convex} for general games and \cite{parise2019variational} for network games). Specifically, under Assumption \ref{a:strat_set}, $\sbar$ is a Nash equilibrium of the game played over the network $\tilde{G}$ if and only if it solves the variational inequality VI$(F(\cdot,\Gtild), \mathcal{S})$, that is, 
\begin{align}
    (s-\sbar)^\top F(\sbar,\Gtild) \ge 0 \qquad \forall s \in \Scal.
\end{align} 
If, additionally, $F(\cdot,\Gtild)$ is strongly monotone, then the equilibrium is unique \cite[Equations (13) and (18)]{scutari2010convex}. 

\subsection{Learning in static repeated games}

In many settings, agents do not know a priori what strategy to play. 
Instead, they may refine their strategy over multiple repetitions of the same game. For example, if the network $\Gtild$ doesn't change over time, agents may use the projected gradient dynamics
\begin{align} \label{eq:grad_dyn}
    s^{k+1} = \Pi_{\mathcal{S}} [ s^k - \tau  F(s^k, \Gtild) ],
\end{align}
where $\tau>0$ is a fixed step size and $s^k$ is the strategy profile at iteration $k\ge 0$, initialized at $s^0 \in \Scal$. 

Under suitable assumptions (including strong monotonicity and Lipschitz continuity of $F(\cdot, \Gtild)$, as well as a small enough step size $\tau$), 
the dynamics in~\eqref{eq:grad_dyn} converge to the unique Nash equilibrium of the game played over the network $\Gtild$ \cite[Theorem~12.1.2]{facchinei2003finite}.

In the next section, we examine a similar learning setting, where, however, the network and the population may change over time. 
Importantly, in such a non-stationary environment, the equilibrium of the underlying stage game is also changing over time, making the learning task non-trivial.

\section{Time-Varying Network Games} \label{sec:dyn_net}

In contrast to the static environment presented in the previous section, real-world interactions are dynamic and frequently changing. Moreover, agents may occasionally abstain from participating in the game. 
This section presents a stochastic model for representing such non-stationary environments, and explores the learning dynamics of agents in such settings.

\subsection{Random network model} \label{sec:rdn_net_model}

We consider a repeated network game where each stage game is played over a random network.
At every repetition $k=0,1,2,\dots$, a new network realization $G^k \in [0,1]^{N\times N}$ is generated by independently sampling the links $G_{ij}^k$ from a distribution with bounded support in the unit interval\footnote{We restrict the support of $G^k_{ij}$ to $[0,1]$ without loss of generality. The results in this paper can easily be extended to allow any bounded convex support.} and expectation $\Gbar_{ij} \in [0,1]$, for all $i,j\in\Ncal$ such that $i\neq j$.
Additionally, the diagonal entries are set to zero to avoid self-loops ($G_{ii}^k=0$ for all $i\in\Ncal$).
By setting $\Gbar_{ii}=0$ for all $i\in\Ncal$ as well, the adjacency matrix $\Gbar=\Eset[G^k]$ represents the expected network of interactions when all agents are present.

Next, we model the random participation of agents in the game.
For each agent $i\in\Ncal$, the Bernoulli random variable $P_{ii}^k \sim \textrm{Ber}(\Pbar_{ii})$ takes a value of one with probability $\Pbar_{ii}>0$ if agent $i$ is participating in the stage game at iteration $k$, and zero otherwise.
We define $P^k \in\{0,1\}^{N\times N}$ and $\Pbar \in[0,1]^{N\times N}$ to be the diagonal matrices with $P_{ii}^k$ and $\Pbar_{ii}$ as their $i$-th diagonal entries, respectively. We refer to $P^k$ as the participation matrix. 

The matrix $G^k P^k$ represents the effective realized network of interactions at iteration $k$ since a realization $P_{ii}^k=0$ will set the $i$-th column of the adjacency matrix $G^k$ to zero, representing the fact that agent $i$ does not contribute to the local aggregate $z_j(s|G^kP^k)$ of any other agent $j \in \Ncal$, and thus, does not affect the cost of other agents at this iteration. 
The case where $\Pbar=\Id_N$ corresponds to the setting where the population is static and all agents participate at every iteration almost surely.

\subsection{Learning dynamics}

For each player $i$, we consider projected gradient dynamics of the form
\begin{align} \label{eq:grad_for_dyn_pop}
    s^{k+1}_i\!&=\!\begin{cases}
    \Pi_{\mathcal{S}_i} \Big[ s^k_i - \tau^k \nabla_{s_i}J_i(s_i^k, \zi{s^k}{G^k P^k} ) \Big] & \text{if}\, P_{ii}^k=1,\\
    s_i^k & \text{if}\, P_{ii}^k=0,
    \end{cases}
\end{align}
for all $k \ge 0$, where $s_i^0 \in \Scal_i$ and $\tau^k > 0$ is the step size at iteration~$k$. 
Compared to the dynamics for the stationary environment introduced in \eqref{eq:grad_dyn}, there are three main differences. First, agents only update their strategy when they are present (i.e., when $P_{ii}^k=1$). Second, they use the current gradient information which depends on the realized time-varying network $G^k$ and on the realized participation matrix $P^k$.\footnote{Note that the agents need not know the effective realized network $G^k P^k$ or the full strategy profile $s^k$ to update their strategy. Instead, they only need to know their current strategy $s_i^k$ and their local aggregate $z_i(s^k|G^kP^k)$ sensed via the realized network.}  Third, agents are allowed to use time-varying step-sizes $\tau_k$ to compensate for the non-stationary environment. 

Note that due to the stochastic nature of the network, the gradient update of each agent given in \eqref{eq:grad_for_dyn_pop} is random.
If we view this random update as a deterministic update subject to a random perturbation, we can interpret the gradient dynamics in \eqref{eq:grad_for_dyn_pop} as a form of stochastic gradient play for a deterministic game in which players minimize their expected cost, as formalized next. 

Given a random network $\Gtild\Ptild$ distributed identically to the random network $G^k P^k$ described in Section~\ref{sec:rdn_net_model}, let $\Jtild_i : \Rset^n \to \Rset$ denote the conditional expectation of the cost incurred by agent $i$ given a strategy profile $s\in\Scal$, which is given by
\begin{align} \label{eq:Jtild}
    \Jtild_i(s) := \Eset [J_i(s_i,\zi{s}{\Gtild \Ptild}) \, | \, s]
\end{align}
where the expectation is taken over the  realizations of $\Gtild$ and $\Ptild$.  
The game Jacobian operator of the corresponding game is defined~as
\begin{align}
\label{eq: F tilde}
    \Ftild(s) := [\nabla_{s_i} \Jtild_i(s)]_{i\in\Ncal}.
\end{align}
The game Jacobian operator $\Ftild$ can be shown to be the expectation of the random game Jacobian operator $F(\cdot,\Gtild \Ptild)$.
\begin{lemma}
\label{lem:grad_exp=exp_grad}
    Suppose that Assumption \ref{a:strat_set} holds. For each $i \in \Ncal$, the cost function $\Jtild_i(s)$ is well-defined. Moreover, for all $i \in \Ncal$ and $s \in \Scal$, 
    \begin{align}
        \nabla_{s_i} \Eset [J_i(s_i,\zi{s}{\Gtild \Ptild}) \, | \, s ] 
        &= \Eset [\nabla_{s_i} J_i(s_i,\zi{s}{\Gtild \Ptild}) \, | \, s ] 
    \end{align}
    and
    \begin{align}
        \Big\| \nabla_{s_i} \Eset [J_i(s_i,\zi{s}{\Gtild \Ptild}) \, | \, s ] \Big\|_2 \leq \bar{J}_2. 
    \end{align}\vspace{0.01cm}
\end{lemma}

\begin{proof}
    The desired result follows from the dominated convergence theorem where boundedness of the sequence on which the theorem is applied can be demonstrated via Assumption \ref{a:strat_set} and the mean value theorem.
\end{proof}

In the next result, we show that the learning dynamics in \eqref{eq:grad_for_dyn_pop} can be rewritten as a projected stochastic gradient descent scheme.
\begin{proposition} \label{prop:sgd_dyn}
    The dynamics described in \eqref{eq:grad_for_dyn_pop} are equivalent to the following dynamics
    \begin{align} \label{eq:sgd_dyn}
        s^{k+1} =  \Pi_\Scal \Big[ s^k - \tau^k \Delta ( \Ftild(s^k) + w^k )\Big]
    \end{align}
    for all $k \geq 0$, where $\Delta := \Pbar \otimes \Id_n$ and
    \begin{align} \label{eq:wk_defn}
        w^k := \Delta^{-1} (P^k \otimes \Id_n) F(s^k,G^kP^k) - \Ftild(s^k).  
    \end{align}
\end{proposition}
\begin{proof}
    The strategy update of player $i$ in \eqref{eq:grad_for_dyn_pop} can be rewritten as follows
    \begin{align*}
        s_i^{k+1} &= \Pi_{\Scal_i} \Big[ s_i^k - P_{ii}^k \cdot \tau^k \nabla_{s_i}J_i(s_i^k, \zi{s^k}{G^k P^k} )  \Big]\\
        &= \Pi_{\Scal_i} \Big[ s_i^k - (\bar{P}_{ii} \tau^k) \Big(  \frac{P_{ii}^k}{\bar{P}_{ii}}  \nabla_{s_i}J_i(s_i^k, \zi{s^k}{G^k P^k}) \\
        &\quad + \nabla_{s_i}\Jtild(s^k) - \nabla_{s_i}\Jtild(s^k) \Big )\Big]\\
        &= \Pi_{\Scal_i} \Big[ s_i^k - (\bar{P}_{ii} \tau^k) ( \nabla_{s_i}\Jtild(s^k) + w_i^k )\Big]
    \end{align*}
    where $w_i^k:= (P_{ii}^k/\bar{P}_{ii}) \cdot \nabla_{s_i}J_i(s_i^k, \zi{s^k}{G^k P^k})-\nabla_{s_i}\Jtild_i(s^k)$.
\end{proof}

The perturbation vector $w^k$ can be shown to satisfy some properties that will be useful for the convergence analysis of the dynamics in~\eqref{eq:sgd_dyn}.

\begin{proposition} \label{prop:wk}
    Suppose that Assumption \ref{a:strat_set} holds. Let $\{\Fcal_k\}_{k\ge 0}$ denote an increasing sequence of $\sigma$-algebras such that $s^k$ is $\Fcal_k$-measurable.
    Then, for all $k \geq 0$, the vector $w^k$ defined in \eqref{eq:wk_defn} satisfies  
    \begin{align}
        \label{eq:wk_cond_1} 
        \Eset[\, w^k \, \vert\,  \Fcal_k\, ] &= 0, 
    \end{align}
    and
    \begin{align}
        \label{eq:wk_cond_2} 
        \Eset[\,  \|w^k\|_2^2 \, \vert\,  \Fcal_k\, ] &\le M^2 N < \infty
    \end{align}where $M:=\bar{J}_2 ( (\min_{i\in\Ncal} \Pbar_{ii})^{-1} + 1)$.
\end{proposition}

\begingroup
\allowdisplaybreaks
\begin{proof}
For each $i \in \Ncal$, we have
\begin{align*}
    &\Eset[w^k_i \ | \ \Fcal_k] \\
    &= \Eset\left[ \frac{P_{ii}^k}{\bar{P}_{ii}} \cdot \nabla_{s_i}J_i(s_i^k, \zi{s^k}{G^k P^k}) - \nabla_{s_i}\Jtild_i(s^k) \ \Bigg| \ s^k\right] \\
    &\overset{(a)}{=} \frac{\Eset\left[P_{ii}^k\right]}{\bar{P}_{ii}} \cdot \Eset \left[\nabla_{s_i}J_i(s_i^k, \zi{s^k}{G^k P^k}) \ | \ s^k\right] - \nabla_{s_i}\Jtild_i(s^k) \\
    &\overset{(b)}{=} \frac{\bar{P}_{ii}}{\bar{P}_{ii}} \cdot \nabla_{s_i} \Eset \left[J_i(s_i^k, \zi{s^k}{G^k P^k}) \ | \ s^k\right] - \nabla_{s_i}\Jtild_i(s^k) \\
    &=  \nabla_{s_i} \Jtild_i(s^k) - \nabla_{s_i}\Jtild_i(s^k) = 0
\end{align*}
where $(a)$ follows from the fact that $P_{ii}^k$ is independent from $G_{ij}^k P_{jj}^k$ for all $j\neq i$ and that $G_{ii}^k P_{ii}^k = 0$, and $(b)$ follows from Lemma \ref{lem:grad_exp=exp_grad}. 
For any $i \in \Ncal$, we also have
\begin{align*}
    \|w^k_i\|_2 &= \left\|\frac{P_{ii}^k}{\bar{P}_{ii}} \cdot \nabla_{s_i}J_i(s_i^k, \zi{s^k}{G^k P^k}) - \nabla_{s_i}\Jtild_i(s^k)\right\|_2 \\
    &\leq \frac{P_{ii}^k}{\bar{P}_{ii}} \cdot \|\nabla_{s_i}J_i(s_i^k, \zi{s^k}{G^k P^k})\|_2 + \|\nabla_{s_i}\Jtild_i(s^k)\|_2 \\
    &\leq \frac{1}{\min_{i \in \Ncal} \Pbar_{ii}} \bar{J}_2 + \bar{J}_2 = M
\end{align*} 
where the second inequality follows from Assumption \ref{a:strat_set} and Lemma \ref{lem:grad_exp=exp_grad}. Hence, 
\begin{align*}
    \Eset[\, \|w^k \|^2_2 \, | \, \Fcal_k \,] &= \Eset\left[ \sum_{i \in \Ncal} \|w^k_i \|^2_2 \ \Bigg | \ \Fcal_k\right] \\
    &\leq \Eset \left[\sum_{i \in \Ncal} M^2 \ \Bigg| \ \Fcal_k \right] =M^2 N. 
\end{align*}
\end{proof}
\endgroup

\section{Convergence and Regret Analysis} \label{sec:conv_regr}

In this section, we examine the convergence properties of the learning dynamics presented in \eqref{eq:grad_for_dyn_pop} (equivalently, \eqref{eq:sgd_dyn}) and the regret incurred by each agent when following these dynamics.

\subsection{Convergence}

To compensate for the network variability, we first assume that each agent uses a diminishing step size.

\begin{assumption}[Step size] \label{a:step_size} 
    The step size sequence $\{\tau^k\}_{k \geq 0}$ satisfies $\tau^k>0$ for all $k$, $\tau^k\to 0$ as $k \rightarrow \infty$, $\sum_{k=0}^\infty \tau^k = \infty$ and $\sum_{k=0}^\infty (\tau^k)^2 < \infty$.
\end{assumption}

To analyze the asymptotic behavior of the stochastic dynamics in \eqref{eq:grad_for_dyn_pop}, we follow the same approach as \cite{jiang2008stochastic} and show that we can express the sequence of deviations between the strategy profile iterates $s^k$ and a candidate limiting profile as an \textit{almost supermartingale} sequence \cite{robbins1971convergence}.
This allows us to leverage the following known result from \cite{robbins1971convergence} which demonstrates almost sure asymptotic convergence of almost supermartingale sequences.

\begin{lemma}[\hspace{1sp}{\cite{robbins1971convergence}}] \label{lem:supermartingale}
    Let $\{\Fcal_k\}_{k\ge 0}$ be an increasing sequence of $\sigma$-algebras.
    For each $k=0,1,2,\dots$, let $V_k$, $\alpha_k$, $\beta_k$ and $\gamma_k$ be nonnegative random variables adapted to $\Fcal_k$.
    If it holds almost surely that $\Sigma_{k=0}^\infty \alpha_k<\infty$, $\Sigma_{k=0}^\infty \beta_k<\infty$ and 
    \begin{align*}
        \Eset[V_{k+1} \mid \Fcal_k] \le (1+\alpha_k) V_k + \beta_k - \gamma_k,
    \end{align*}
    then $\{V_k\}_{k\ge 0}$ is convergent almost surely and $\Sigma_{k=0}^\infty \gamma_k < \infty$ almost surely.
\end{lemma} 

With this lemma and with strong monotonicity of the operator $\Ftild$, we show that the dynamics in \eqref{eq:grad_for_dyn_pop}  converge almost surely.
\begin{assumption}[Monotonicity of $\Ftild$] \label{a:Ls_mu}
The operator $\Ftild$ is $\mu$-strongly monotone for some $\mu>0$, i.e., $(s_1-s_2)^\top(\Ftild(s_1)-\Ftild(s_2))\ge \mu \|s_1-s_2\|_2^2$ for all $s_1,s_2\in\Rset^{nN}$.
\end{assumption}

\begin{theorem} \label{thm:conv} 
    Suppose that Assumptions \ref{a:strat_set}, \ref{a:step_size} and \ref{a:Ls_mu} hold. 
    Given $s^0\in\Scal$, let $\{s^k\}_{k\ge0}$ be a sequence generated by the iteration in \eqref{eq:grad_for_dyn_pop} and let $\sbar$ denote the unique solution to VI$(\Ftild,\Scal)$. 
    Then, $s^k$ converges to $\sbar$ almost surely.
\end{theorem}

\begin{proof}    
    To demonstrate almost sure convergence, we start by bounding the expected deviation between $s^k$ and $\sbar$ in the norm weighted\footnote{For a given positive definite matrix $Q \in\Rset^{nN\times nN}$, we define the weighted norm $\| s \|_Q := \sqrt{s^T Q s}$.} by $\Delta^{-1}$.
    For any $k \geq 0$, we have
    \begin{align*}
        &\Eset[\|s^{k+1}-\sbar\|_{\Delta^{-1}}^2 \,\vert\, \Fcal_k] \\
        &\overset{(a)}{=} \Eset[ \| \Pi_\Scal[s^k - \tau^k \Delta (\Ftild(s^k)+w^k)]- \Pi_\Scal[\sbar] \|_{\Delta^{-1}}^2 \,\vert\, \Fcal_k] \\
        &\overset{(b)}{\le} \Eset[\|s^k - \tau^k \Delta (\Ftild(s^k)+w^k) - \sbar \|_{\Delta^{-1}}^2 \,\vert\, \Fcal_k] \\
        &= \|s^k-\sbar\|_{\Delta^{-1}}^2 - 2 \tau^k (s^k-\sbar)^\top \Delta^{-1} \Delta (\Ftild(s^k) +\Eset[w^k \,\vert\, \Fcal_k] ) \\
        &\quad + (\tau^k)^2 \left(\| \Delta \Ftild(s^k)\|_{\Delta^{-1}}^2 + \Eset[\|\Delta w^k\|_{\Delta^{-1}}^2 \,\vert\, \Fcal_k]\right) \\
        &\quad + 2 (\tau^k)^2 (\Delta \Ftild(s^k))^\top \Delta^{-1} ( \Delta \Eset[w^k \,\vert\, \Fcal_k]) \\
        &\overset{(c)}{=} \|s^k-\sbar\|_{\Delta^{-1}}^2 - 2 \tau^k (s^k-\sbar)^\top \Ftild(s^k) \\
        &\quad + (\tau^k)^2 \left( \| \Delta \Ftild(s^k)\|_{\Delta^{-1}}^2 + \Eset[\|\Delta w^k\|_{\Delta^{-1}}^2 \,\vert\, \Fcal_k]\right) \\
        &\overset{(d)}{\le} \underbrace{\|s^k-\sbar\|_{\Delta^{-1}}^2}_{=: V_k}  - \underbrace{2 \tau^k (s^k-\sbar)^\top \Ftild(s^k)}_{=:\gamma_k}\\
        &\quad  + \underbrace{(\tau^k)^2 \max_{i\in\Ncal} \Pbar_{ii}  (\bar{J}_2^2 N + M^2 N)}_{=: \beta_k} , 
    \end{align*}
    where
    $(a)$ follows from Proposition \ref{prop:sgd_dyn},
    $(b)$ follows from the fact that the projection operator $\Pi_\Scal$ is nonexpansive in the weighted norm,\footnote{The projection under the standard and weighted norm coincide due to the structure of $\Delta^{-1}$.} 
    $(c)$ follows from the fact that $\Eset[w^k \,\vert\, \Fcal_k]=0$ by Proposition \ref{prop:wk}, 
    and $(d)$ follows from the fact that for any vector $s \in \Rset^{nN}$, $\|\Delta s\|_{\Delta^{-1}}^2 = s^T \Delta \Delta^{-1} \Delta s = s^T \Delta s \leq \max_{i\in\Ncal} \Delta_{ii} \|s\|_2^2$, and from Lemma \ref{lem:grad_exp=exp_grad} and Proposition \ref{prop:wk}. 

    In the following, we show that the assumptions of Lemma~\ref{lem:supermartingale} are satisfied by the sequence $\{V_k\}_{k\ge 0}$, with $\alpha_k=0$ for all $k$.
    By Assumption \ref{a:step_size}, $\tau^k$ is square-summable which implies that $\sum_{k = 0}^\infty \beta_k < \infty$. 
    To show that $\gamma_k$ is nonnegative, note that
    \begin{align}
        \gamma_k &=  2 \tau^k (s^k-\sbar)^\top \Ftild(s^k) \nonumber \\ 
        &= 2 \tau^k [(s^k-\sbar)^\top (\Ftild(s^k)-\Ftild(\sbar)) + (s^k-\sbar)^\top \Ftild(\sbar) ] \nonumber \\ 
        &\ge 2 \tau^k \mu \|s^k-\sbar\|_2^2, \label{eq: gamma_k lower bound}
    \end{align}
    where the inequality follows from strong monotonicity of $\Ftild$ (Assumption \ref{a:Ls_mu}) and from the fact that $\sbar$ solves the variational inequality  $(s-\sbar)^\top \Ftild(\sbar) \ge 0 \ \forall s\in\Scal$. 
    Hence, by Lemma \ref{lem:supermartingale}, there exists a random variable $\epsilon\ge 0$ such that $\|s^k - \sbar\|_{\Delta^{-1}}^2 \to \epsilon$ as $k \rightarrow \infty$ almost surely, and $\sum_{k = 0}^\infty \gamma_k < \infty$ almost surely. 

    Let $\Omega$ be the sample space, and let $\epsilon : \Omega \to \Rset_+$, and for each $k \geq 0$, $\gamma_k : \Omega \to \Rset_+$. Suppose there exists a positive measure set $A \subseteq \Omega$ such that for all $\omega \in A$, $\epsilon(\omega) > 0$. Then, for all $\omega \in A$, there exists some $\Bar{k}(\omega) \geq 0$ such that for all $k \geq \Bar{k}(\omega)$, $\|s^k - \Bar{s}\|_{\Delta^{-1}}^2 \geq \epsilon(\omega)/2$. Further, using \eqref{eq: gamma_k lower bound}, we obtain
    \begin{align}
        \sum_{k = 0}^\infty \gamma_k(\omega) \geq \sum_{k = \bar k(\omega)}^\infty \gamma_k(\omega) \geq \mu \epsilon(\omega) \cdot \sum_{k = \bar k(\omega)}^\infty \tau^k = \infty
    \end{align}
    since $\tau^k$ is non-negative and non-summable by Assumption~\ref{a:step_size}. Thus, for all $\omega \in A$, $\sum_{k = 0}^\infty \gamma_k(\omega) = \infty$, resulting in a contradiction to the fact that $\sum_{k = 0}^\infty \gamma_k < \infty$ almost surely. 
    Therefore, we must have $\epsilon = 0$ almost surely, which implies that $\{s^k\}_{k\ge 0}$ converges to $\sbar$ almost surely. 
\end{proof}

The following corollary is an immediate consequence of Theorem \ref{thm:conv} and provides a connection between the strategy learned with the gradient dynamics in the time-varying game and the unique Nash equilibrium of the static game played over the expected cost functions. 
\begin{corollary}
    Suppose that Assumptions \ref{a:strat_set}, \ref{a:step_size} and \ref{a:Ls_mu} hold. 
    Given $s_0\in\Scal$, let $\{s^k\}_{k\ge0}$ be a sequence generated by the iteration in \eqref{eq:grad_for_dyn_pop}. 
    Then, $s^k$ converges to the unique Nash equilibrium of the game in which players minimize the expected cost functions \eqref{eq:Jtild} almost surely.
\end{corollary}
\begin{proof}
    By Theorem \ref{thm:conv}, $s^k$ converges almost surely to the unique solution $\sbar$ of the variational inequality $(s-\sbar)^\top \Ftild(\sbar) \ge 0 \ \forall s\in\Scal$, which is the unique Nash equilibrium of the game played over the expected cost functions \eqref{eq:Jtild}. 
\end{proof}

To help derive bounds on the regret incurred by each player when playing the strategy prescribed by each iterate $s^k$, we next derive rates for the almost sure convergence of $s^k$ to $\sbar$ and for the convergence in the mean-square sense, under different choices of step size sequences. 

\begin{theorem} \label{thm:rates}
    Suppose that Assumptions \ref{a:strat_set}, \ref{a:step_size} and \ref{a:Ls_mu} hold. 
    \begin{enumerate}
        \item Let the step size sequence be such that $\tau^0 = 1$ and 
        \begin{align} 
            \tau^k = \frac{1}{k^{1 - \theta}}
        \end{align}
        for all $k \geq 1$ and for some fixed $\theta \in (0,1/2)$.  
        Then, for any fixed $N \geq 1$ and for any $\beta \in (2\theta, 1)$, the following holds almost surely\footnote{Eq. \eqref{eq:as_conv} means that, for a fixed $N$, $\lim_{k \rightarrow \infty} \|s^k - \sbar\|_2 \cdot k ^{(1 - \beta)/2} = 0$.} 
        \begin{align} \label{eq:as_conv}
            \|s^k-\sbar\|_2 = o\left(\frac{1}{k^{(1 - \beta)/2}}\right). 
        \end{align}
        
        \item Let $C_2 := 2 \mu \min_{i \in \Ncal} \Pbar_{ii}$. 
        Let the step size sequence be such that 
        \begin{align}
            \tau^k = B \delta_k,
        \end{align}
         with $ B > 1/C_2$ and where $\{\delta_k\}_{k \geq 1}$ is a chosen sequence that satisfies
        \begin{align}
            \delta_k\rightarrow 0, \quad  \delta_k (1 - \delta_k) \leq \delta_{k+1}
        \end{align}
        for all $k \geq 1$.  
        Then, $s^k$ converges to $\sbar$ in the mean square sense with the following rate of convergence:
        \begin{align}
            \Eset[\|s^k-\sbar\|_2] \leq \sqrt{N D \delta_k} \quad \text{for all } k\ge K
        \end{align}
        for some $D>0$ and $K\geq1$ large enough.
    \end{enumerate}
\end{theorem}

\begin{proof}
    \begin{enumerate}
        \item Let $C_1 := \max_{i \in \Ncal} \Pbar_{ii} (\bar{J}_2^2 + M^2)$. 
        By following the steps of
        the proof of Theorem \ref{thm:conv}, we obtain that for all $k \geq 0$,
        \begin{align*}
            &\Eset[\|s^{k+1}-\sbar\|_{\Delta^{-1}}^2 \,\vert\, \Fcal_k] \leq \|s^k-\sbar\|_{\Delta^{-1}}^2 \\
            &\qquad + (\tau^k)^2 \max_{i\in\Ncal} \Pbar_{ii} (\bar{J}_2^2 + M^2) N - 2 \tau^k (s^k-\sbar)^\top \Ftild(s^k) \\
            &\leq \|s^k-\sbar\|_{\Delta^{-1}}^2 + (\tau^k)^2 C_1 N - 2 \tau^k  \mu \|s^k - \sbar\|_2^2. 
        \end{align*}
        Note that for any vector $s \in \Rset^{Nn}$, we have $\| s \|_2^2 \geq \min_{i \in \Ncal} \Pbar_{ii} \| s \|_{\Delta^{-1}}^2$.
        Hence, 
        \begin{align} 
            \nonumber&\Eset[\|s^{k+1}-\sbar\|_{\Delta^{-1}}^2 \,\vert\, \Fcal_k] \\
            \nonumber&\leq \|s^k-\sbar\|_{\Delta^{-1}}^2 + (\tau^k)^2 C_1 N - 2 \tau^k \mu \min_{i \in \Ncal} \Pbar_{ii} \|s^k - \sbar\|_{\Delta^{-1}}^2 \\
            &= (1 - \tau^k C_2) \|s^k - \sbar\|_{\Delta^{-1}}^2 + (\tau^k)^2 C_1 N. \label{eq: expected error inequality}
        \end{align} 
        Hence, by \cite[Lemma 1]{liu2022almost}, since $\theta \in (0,1/2)$ and $\tau^k=1/k^{1-\theta}$ for $k\ge 1$, it follows that for any fixed $\beta \in (2\theta, 1)$, 
        \begin{align*}
            \|s^k-\sbar\|_{\Delta^{-1}}^2 = o\left(\frac{1}{k^{1 - \beta}}\right) 
        \end{align*}
        for all $k \geq 0$ almost surely. Hence, \eqref{eq:as_conv} follows.
        
        \item By taking the expectation on both sides of \eqref{eq: expected error inequality} and defining $e_k := \Eset\left[\|s^k - \sbar\|_{\Delta^{-1}}^2\right]$, we obtain that for all $k \geq 0$,
        \begin{align*}
            e_{k+1} \leq (1 - \tau^k C_2) e_k + (\tau^k)^2 C_1 N. 
        \end{align*}
        Let $C_3 := 4 \smax^2/\min_{i\in\Ncal} \Pbar_{ii}$. Note that for all $k \geq 0$, 
        \begin{align*}
            \|s^k - \sbar\|_{\Delta^{-1}}^2 &\leq \frac{1}{\min_{i\in\Ncal} \Pbar_{ii}} \|s^k - \sbar\|_2^2 \\
            &\leq \frac{1}{\min_{i\in\Ncal} \Pbar_{ii}} (\|s^k\|_2 + \|\sbar\|_2)^2\\
            &\leq \frac{4 \smax^2 N}{\min_{i\in\Ncal} \Pbar_{ii}} \\
            &= C_3 N, 
        \end{align*}
        where the third inequality is due to Assumption \ref{a:strat_set}.
        Since $\delta_k\rightarrow 0$, there exists $K>0$ such that $\delta_k \le 1/(BC_2)$ for all $k\ge K$. Let
        \begin{align}
            D =  \max\left\{\frac{C_3}{\delta_K}, \ \frac{B^2 C_1}{B C_2 - 1}\right\}.  
        \end{align}
        We show by induction that $e_k \leq N D \delta_k$ for all $k \geq K$.
        For the base case $k=K$, we have that
        \begin{align*}
            e_K = \Eset[\|s^K - \sbar\|_{\Delta^{-1}}^2] \leq C_3 N \leq N D \delta_K. 
        \end{align*}
        Suppose now that $e_k \le N D \delta_k$ holds for some $k\geq K$. Then,
        \begin{align*}
            e_{k+1} &\leq (1 - \tau^{k} C_2) e_{k} + (\tau^{k})^2 C_1 N \\
            &\overset{(a)}{\leq} (1 - B \delta_{k} C_2) N D \delta_{k} + (B \delta_{k})^2 C_1 N \\
            &= N D \delta_{k} + (\delta_{k})^2N (- D B C_2 + B^2 C_1) \\
            &= N D \delta_{k} - N D (\delta_{k})^2 \\
            &\qquad - (\delta_{k})^2 N (B C_2 - 1) \cdot \left(D - \frac{B^2 C_1}{B C_2 - 1}\right) \\
            &\overset{(b)}{\leq} N D \delta_{k}( 1 - \delta_{k}) \\
            &\overset{(c)}{\leq} N D \delta_{k+1}.  
        \end{align*}
        where $(a)$ follows from $\tau^k=B\delta_k$, from the fact that $B C_2 \delta_k \leq 1 $, and from $e_k \le N D \delta_k$, $(b)$ follows from $B C_2 > 1$ and $D \geq B^2 C_1/(B C_2 - 1)$, and $(c)$ follows from $\delta_k(1 - \delta_k) \leq \delta_{k+1}$. 
        
         Thus, we obtain that for all $k \geq K$, 
        \begin{align*}
            (\Eset[\|s^k - \sbar\|_2])^2 &\leq 
            \Eset[\|s^k - \sbar\|_2^2] \\
            &\leq \Eset[\|s^k - \sbar\|_{\Delta^{-1}}^2] 
            \leq N D \delta_k.  
        \end{align*}        
    \end{enumerate}
\end{proof}

The following corollary gives an example of a step size sequence that satisfies the conditions of Theorem \ref{thm:rates}, Part 2). 
\begin{corollary}
\label{cor:delta_k_choice}
    Suppose that Assumptions \ref{a:strat_set}, \ref{a:step_size} and \ref{a:Ls_mu} hold. 
    Let $\tau^0=B$ and $\tau^k = B/k^\alpha$ for all $k\geq 1$ where $B=2^\alpha/C_2$ and $\alpha\in(0,1]$. 
    Then, it holds that $s^k$ converges to $\sbar$ in the mean square sense with the following rate of convergence:
    \begin{align}
        \Eset[\|s^k-\sbar\|_2] \leq \sqrt{\frac{N D}{k^\alpha}} \quad \text{for all } k\ge 2,
    \end{align}
    where $D = \max\left\{ 2^\alpha C_3, B^2 C_1/(2^\alpha - 1)\right\}$ with $C_3=4\smax^2/\min_{i\in\Ncal}\Pbar_{ii}$. 
\end{corollary}
\begin{proof}
    First, we show that the sequence $\delta_k=1/k^\alpha$ satisfies the condition $\delta_{k+1} - \delta_k(1 - \delta_k) \geq 0$ for all $k$.  
    Note that for all $k \geq 1$, we have
    \begin{align*}
         \delta_{k+1} - \delta_k  + \delta_k^2  &= \frac{1}{(k+1)^\alpha} - \frac{1}{k^\alpha} + \frac{1}{k^{2\alpha}}\\
         &=\frac{k^{2\alpha} - k^\alpha(k+1)^\alpha + (k+1)^\alpha}{k^{2\alpha}(k+1)^\alpha}\\
         &\ge \frac{k^{2\alpha} - k^\alpha(k^\alpha+1) + (k+1)^\alpha}{k^{2\alpha}(k+1)^\alpha}\\
         &= \frac{(k+1)^\alpha - k^\alpha}{k^{2\alpha}(k+1)^\alpha}\\
         &\ge 0
    \end{align*}
    where the first inequality follows from Lemma \ref{lem:k_alpha} in the appendix, and the last inequality follows from the fact that $x^\alpha$ is a nondecreasing function of $x$, so $(k+1)^\alpha \ge k^\alpha$.
    Finally, note that for any $\alpha \in (0,1]$, $BC_2 = 2^\alpha > 1$.
    Hence, by Theorem \ref{thm:rates} (with $K=2$), the desired result follows.
\end{proof}

\subsection{Regret guarantees}

In this section, we justify why the Nash equilibrium $\sbar$ of the game played over the expected costs is a good strategy to learn by demonstrating that it is an $\epsilon$-Nash equilibrium of each stage game with high probability, for a large enough population.
Moreover, we derive bounds on the individual regret incurred by an agent when following the proposed learning dynamics. 
We define the instantaneous regret of an agent $i$ at stage $k$ as the excess cost incurred for not playing their best response strategy in stage  $k$. 
More precisely, given any $k \geq 0$, $i \in \Ncal$ and a strategy profile $s^k$, the instantaneous individual regret of agent $i$ is given by
\begin{align}
\label{eq: regret definition}
    &R_i(s_i^k, \zi{s^k}{G^k P^k}) \nonumber \\
    &:= J_i(s_i^k, \zi{s^k}{G^k P^k}) - \inf_{s_i\in\Scal_i} J_i(s_i, \zi{s^k}{G^k P^k}).
\end{align}

To prove that $\sbar$ results in an $\epsilon$-Nash equilibrium, we use a bound on the distance between the local aggregates generated by a realized network and the expected network, derived as an intermediate step in the proof of Corollary 1 in \cite{altaha2023gradient}.

\begin{lemma}[\hspace{1sp}\cite{altaha2023gradient}]
\label{lem:local_agg_bound}
Suppose that Assumption \ref{a:strat_set} holds. Let $s \in \Scal$ and let $G P$ be a random network identically distributed to the random network $G^k P^k$ described in Section \ref{sec:rdn_net_model}.
Then, for any fixed agent $i \in \Ncal$ and $\delta > 0$, 
    \begin{align*}
        \|z_i(s|GP) - z_i(s|\Gbar \Pbar)\|_2 \leq \sqrt{\frac{n\smax^2\log(4nN/\delta)}{2N}}
    \end{align*} 
with probability at least $1 - \delta/(2N)$. 
\end{lemma}

Using the lemma above, we prove that the strategy learned by the agents through the dynamics in  \eqref{eq:grad_for_dyn_pop} is an $\epsilon$-Nash equilibrium of each stage game with high probability, where $\epsilon$ decreases with the number of agents $N$. 
To prove the result, we introduce a smoothness assumption on the cost functions of the agents. 

\begin{assumption}[Smoothness of $J_i$]
\label{a:smooth-LJz}
    For all agents $i \in \Ncal$, the cost function $J_i(s,z)$ is Lipschitz continuous in both $s$ and $z$ with constants $L_J^s$, $L_J^z$, respectively, i.e.,
    \begin{align*}
        |J_i(s,z) - J_i(\tilde{s},\Tilde{z})| \leq L^s_J \|s - \Tilde{s}\|_2 + L^z_J \|z - \Tilde{z}\|_2 
    \end{align*}
    for all $s, \tilde{s}, z, \tilde{z}$. 
\end{assumption}

\begin{proposition}
\label{prop:epsilon_Nash}
Suppose that Assumptions \ref{a:strat_set}, \ref{a:step_size}, \ref{a:Ls_mu} and \ref{a:smooth-LJz} hold. Fix any $N \geq 1$, $k \geq 0$ and $\delta > 0$. 
Then, the Nash equilibrium $\sbar$ of the game played over the expected cost functions $\{\Tilde{J}_i(s)\}_{i \in \Ncal}$ is an $\epsilon^k_{N,\delta}$-Nash equilibrium of the stage game played over the realized cost functions $\{J_i(s_i,z_i(s|G^kP^k))\}_{i \in \Ncal}$, where $\epsilon_{N,\delta}^k$ is a random variable with the following distribution
\begin{align*} 
    \epsilon_{N,\delta}^k =
    \begin{cases}
        \bar{\epsilon}_{N,\delta} \ &\text{w.p. } 1 - \delta, \\
        4 L^z_J \smax \ &\text{w.p. } \delta, 
    \end{cases} 
\end{align*}
and where
\begin{align*}
    \bar{\epsilon}_{N,\delta} &:= 4 L^z_J \left( \sqrt{\frac{n\smax^2\log(4nN/\delta)}{2N}} \left(2 - \frac{\delta}{2N}\right)+ \frac{\delta\smax}{N} \right). 
\end{align*}
\end{proposition}

Note that the superscript in $\epsilon^k_{N,\delta}$ emphasizes the fact that $\epsilon_{N,\delta}^k$ is a random variable dependent on the realization of the network at iteration $k$.

\begin{proof}
    Let $G P$ be a random matrix identically distributed to the random matrix $G^k P^k$ described in Section \ref{sec:rdn_net_model}. 
    To streamline notation, let $\zbar_i(GP) := \zi{\sbar}{GP}$.
    Moreover, for clarity, we will use the notation $\Eset_{G,P}[\cdot]$ to refer to an expectation that is taken with respect to the random matrices $G$ and $P$ only.
    Fix any $i \in \Ncal$.  
    To show the desired result, we bound the deviation in agent $i$'s cost when playing $\sbar_i$ as opposed to any other strategy $\tilde{s}_i\in\Scal_i$ as follows 
    \begin{align}
        &J_i(\bar{s}_i,\zbar_i(G^kP^k)) - J_i(\tilde{s}_i,\zbar_i(G^kP^k)) \nonumber \\
        &= J_i(\bar{s}_i,\zbar_i(G^kP^k)) - \Eset_{G,P}[ J_i(\bar{s}_i,\zbar_i(GP)) ]\nonumber\\
        &\qquad + \underbrace{\Eset_{G,P}[J_i(\bar{s}_i,\zbar_i(GP)] - \Eset_{G,P}[ J_i(\tilde{s}_i,\zbar_i(GP))]}_{\leq 0} \nonumber \\
        &\qquad + \Eset_{G,P}[J_i(\tilde{s}_i,\zbar_i(GP)) ] - J_i(\tilde{s}_i,\zbar_i(G^kP^k)) \nonumber \\
        &\le \Eset_{G,P}[J_i(\bar{s}_i,\zbar_i(GP)) - J_i(\bar{s}_i,\zbar_i(G^kP^k)) \,\vert\, G^k,P^k]\nonumber\\
        &\qquad + \Eset_{G,P}[J_i(\tilde{s}_i,\zbar_i(GP)) - J_i(\tilde{s}_i,\zbar_i(G^kP^k))\,\vert\, G^k,P^k]
        \label{eq:J_i_2_terms}
    \end{align}
    where the inequality follows from the fact that $\bar{s}$ is the Nash equilibrium of the game played over the expected cost functions.
    We upper bound each of the two terms above separately.
    For the first term, we have
    \begin{equation}\label{eq:step}
    \begin{aligned}
        &\Eset_{G,P}\left[J_i(\bar{s}_i,\zbar_i(G^kP^k)) - J_i(\bar{s}_i,\zbar_i(GP)) \,\vert\, G^k,P^k\right]\\
        &\leq \Eset_{G,P} \left[L^z_J \|\zbar_i(G^kP^k) - \zbar_i(GP)\|_2 \,\vert\, G^k,P^k\right] \\
        &\leq L^z_J \|\zbar_i(G^kP^k) - \zbar_i(\Gbar\Pbar)\|_2 \\
        &\qquad + L^z_J \Eset_{G,P}[\|\zbar_i(GP) - \zbar_i(\Gbar\Pbar)\|_2 ]
    \end{aligned} 
       \end{equation}
    where the first inequality follows from Assumption \ref{a:smooth-LJz}.
    By Lemma \ref{lem:local_agg_bound}, for any $\delta>0$, we have that with probability at least $1-\delta/2N$
    \begin{align*}
        \|\zbar_i(GP) - \zbar_i(\Gbar \Pbar)\|_2 \leq \sqrt{\frac{n\smax^2\log(4nN/\delta)}{2N}} =: C_\delta.
    \end{align*} 
    Note that the same bound holds for $G^k, P^k$ in place of $G, P$ respectively. 
    Further, for all realizations of $G$ and $P$, we have the following \textit{worst-case} upper bound:
    \begin{align*}
        \|\zbar_i(GP) - \zbar_i(\Gbar \Pbar)&\|_2 = \Bigg\|\frac{1}{N} \sum_{j = 1}^N (G_{ij} P_{jj} - \Gbar_{ij} \Pbar_{jj})\bar{s}_j\Bigg\|_2 \\
        &\leq \frac{1}{N} \sum_{j = 1}^N (|G_{ij}| |P_{jj}| + |\Gbar_{ij}| |\Pbar_{jj}|)\left\|\bar{s}_j\right\|_2 \\
        &\leq 2 \smax,
    \end{align*}
    where we used Assumption \ref{a:strat_set} and the fact that $G_{ij}, P_{jj}, \Gbar_{ij}, \Pbar_{jj} \in [0,1]$. 
    Thus, with probability at least $1 - \delta/(2N)$, 
    \begin{align*}
        \Eset_{G,P}  [J_i(\bar{s}_i,&\zbar_i(G^kP^k)) - J_i(\bar{s}_i,\zbar_i(GP)) \,\vert\, G^k,P^k  ] \\
        &\leq 2 L^z_J \left(C_\delta + \left(1 - \frac{\delta}{2N}\right) C_\delta  + \left(\frac{\delta}{2N}\right) 2\smax \right)\\
        &= \frac{\bar{\epsilon}_{N,\delta}}{2}.
    \end{align*}
    The same bound holds for the second term in \eqref{eq:J_i_2_terms} with probability at least $1 - \delta/(2N)$.  
    Thus, with probability at least $1 - \delta/N$, 
    \begin{align*}
        &J_i(\bar{s}_i,\zbar_i(G^kP^k)) - J_i(\tilde{s}_i,\zbar_i(G^kP^k)) \le \bar{\epsilon}_{N,\delta}.
    \end{align*}
    Then, by the union bound, this upper bound holds for all agents $i\in\Ncal$ with probability at least $1-\delta$.
    In other words, $\sbar$ is an $\bar{\epsilon}_{N,\delta}$-Nash equilibrium with probability at least $1-\delta$.

    To get the $\epsilon_{N,\delta}^k$ bound stated in the proposition, we utilize the following worst-case bound which holds for all realizations:  
    \begin{align*}
        \|\zbar_i(GP) - \zbar_i(G^kP^k)\|_2 &\leq 2\smax.
    \end{align*}
    Substituting in \eqref{eq:J_i_2_terms} and \eqref{eq:step}, we obtain that with probability~1, 
    \begin{align*}
        J_i(\bar{s}_i,z_i(\bar{s}|G^kP^k)) - J_i(\tilde{s}_i,z_i(\bar{s}|G^kP^k)) \leq 4 L^z_J \smax
    \end{align*}
    for all $ i \in \Ncal$ and for all $ \tilde{s}_i \in \Scal_i$. 
    Thus, $\Bar{s}$ is an $\epsilon_{N,\delta}^k$-Nash equilibrium of the game with cost functions $\{J_i(s_i,z_i(s|G^kP^k))\}_{i \in \Ncal}$. 
\end{proof}

Having established that the strategy learned by the gradient dynamics in \eqref{eq:grad_for_dyn_pop} is an $\epsilon_{N,\delta}^k$-Nash equilibrium of each stage game, we next derive bounds on the instantaneous individual regret of each player.

\begin{theorem}
\label{thm:regret}
    Suppose that Assumptions \ref{a:strat_set}, \ref{a:step_size}, \ref{a:Ls_mu} and \ref{a:smooth-LJz} hold.   
    Let $\{s^k\}_{k \ge 0}$ be a sequence of strategy profiles generated by the gradient dynamics in \eqref{eq:grad_for_dyn_pop}. Fix any $\delta > 0$ and $N \geq 1$. 
    \begin{enumerate}
        \item Let $C_4 := L^s_J + 2 L^z_J$. 
        If the step size sequence is chosen as per Part 1) of Theorem \ref{thm:rates}, then for all $i \in \Ncal$ and for all $k \geq 0$, the instantaneous individual regret satisfies
        \begin{align} \label{eq:inst_reg_bd}
            R_i(s^k_i, z_i(s^k|G^kP^k)) \leq C_4 \|s^k - \sbar\|_2 + \epsilon_{N,\delta}^k
        \end{align}
        where, for a fixed $N$, $\|s^k-\sbar\|_2 = o\left(1/k^{(1 - \beta)/2}\right)$ almost surely for any $\beta \in (2\theta, 1)$, and $\{\epsilon_{N,\delta}^k\}_{k \geq 0}$ are i.i.d. random variables that follow the distribution described in Proposition~\ref{prop:epsilon_Nash}. 
        \item If the step size sequence is chosen as per Corollary \ref{cor:delta_k_choice} with $\alpha=1$, then for all $i \in \Ncal$, for all $k \geq 2$, the expected instantaneous individual regret satisfies
        \begin{align}
            \nonumber &\Eset[R_i(s_i^k, \zi{s^k}{G^k P^k})] \\
            &\le(L^s_J + 2 L^z_J) \sqrt{\frac{N D}{k}} + \Bar{\epsilon}_{N,\delta}(1 - \delta) + 4 L^z_J \smax \delta .
        \end{align}
    \end{enumerate}
\end{theorem}
\begin{proof}
To derive a bound on the regret, we use the facts that
\begin{enumerate}
    \item[(i)] $s^k$ converges to $\bar{s}$ (Theorem \ref{thm:rates}),
    \item[(ii)] $\sbar$ is an $\epsilon_{N,\delta}^k$-Nash equilibrium of the game played over the network $G^kP^k$ (Proposition \ref{prop:epsilon_Nash}). 
\end{enumerate}

Fix any $k \geq 0$,  $\delta > 0$, and $N \geq 1$. To begin, we rewrite the regret as
\begin{align}
    \label{eq:regret_expanded}
    &R_i(s_i^k,z_i(s^k|G^kP^k)) \nonumber \\
    &= J_i(s_i^k, z_i(s^k|G^kP^k)) - \inf_{s_i \in \Scal_i} J_i(s_i, z_i(s^k|G^kP^k)) \nonumber \\
    &= \underbrace{J_i(s_i^k, z_i(s^k|G^kP^k)) - J_i(\sbar_i, z_i(s^k|G^kP^k))}_{=: T_1} \nonumber \\
    &\quad + \underbrace{J_i(\sbar_i, z_i(s^k|G^kP^k)) - J_i(\sbar_i, z_i(\sbar|G^kP^k))}_{=: T_2} \nonumber \\
    &\quad + \underbrace{J_i(\sbar_i, z_i(\sbar|G^kP^k)) - \inf_{\tilde{s}_i \in \Scal_i} J_i(\Tilde{s}_i, z_i(\sbar|G^kP^k))}_{=: T_3} \nonumber \\
    &\quad + \underbrace{\inf_{\tilde{s}_i \in \Scal_i} J_i(\Tilde{s}_i, z_i(\sbar|G^kP^k)) - \inf_{s_i \in \Scal_i} J_i(s_i, z_i(s^k|G^kP^k))}_{=: T_4}. 
\end{align}
We derive a bound on each term above. 

\noindent For $T_1$: By Assumption \ref{a:smooth-LJz}, $J_i(\cdot, z_i(s^k|G^kP^k))$ is $L^s_J$-Lipschitz. Hence,
\begin{align*}
   T_1 \leq |T_1| \leq L^s_J \|s^k_i - \sbar_i\|_2 \leq L^s_J \|s^k - \sbar\|_2. 
\end{align*}

\noindent For $T_2$: By Assumption \ref{a:smooth-LJz}, $J_i(\sbar_i,\cdot)$ is $L^z_J$-Lipschitz. Hence, 
\begin{align*}
    T_2 &\leq |T_2| \\
    &\leq L^z_J \|z_i(s^k|G^kP^k) - z_i(\sbar|G^kP^k)\|_2 \\
    &= L^z_J \left\|\frac{1}{N} [ (G^kP^k \otimes \Id_n) (s^k - \sbar) ] _i \right\|_2 \\
    &\leq \frac{L^z_J}{N} \|G^kP^k\|_2 \|s^k - \sbar\|_2 \\
    &\leq L^z_J \|s^k - \sbar\|_2, 
\end{align*}
where we use the fact that $\|G^k P^k\|_2 \leq \|G^kP^k\|_F \leq N$ since $G^k_{ij}P^k_{jj} \in [0,1]$ for all $i,j$. 

\noindent For $T_3$: By Proposition \ref{prop:epsilon_Nash}, 
\begin{align*}
    T_3 \leq \epsilon_{N,\delta}^k.
\end{align*}

\noindent For $T_4$: We can rewrite the term as
\begin{align*}
    T_4 &= \inf_{\tilde{s}_i \in \Scal_i} \big[\underbrace{J_i(\Tilde{s}_i, z_i(\sbar|G^kP^k)) - J_i(\Tilde{s}_i, z_i(s^k|G^kP^k))}_{=: T_4'} \\
    &\quad + J_i(\Tilde{s}_i, z_i(s^k|G^kP^k))\big] - \inf_{s_i \in \Scal_i} J_i(s_i, z_i(s^k|G^kP^k)). 
\end{align*}
Note that we can bound $T_4'$ in a manner similar to $T_2$ and obtain
\begin{align*}
    T_4' \leq |T_4'| \leq L^z_J \|s^k - \sbar\|_2. 
\end{align*}
Hence, 
\begin{align*}
    T_4  &\leq L^z_J \|s^k - \sbar\|_2 + \inf_{\tilde{s}_i \in \Scal_i} J_i(\Tilde{s}_i, z_i(s^k|G^kP^k)) \\
    &\qquad - \inf_{s_i \in \Scal_i} J_i(s_i, z_i(s^k|G^kP^k)) \\
    &= L^z_J \|s^k - \sbar\|_2. 
\end{align*}

Substituting the bounds derived above into the expression of regret in \eqref{eq:regret_expanded}, we obtain the following upper bound: 
\begin{align}
\label{eq:regret_bound}
    R_i(s^k_i, z_i(s^k|G^kP^k)) \leq
    (L^s_J + 2 L^z_J) \|s^k - \sbar\|_2 + \epsilon_{N,\delta}^k.
\end{align}

To prove the result in part 1) note that for that choice of step size and for fixed $N$, it holds that
\begin{align*}
    \|s^k-\sbar\|_2 = o\left(\frac{1}{k^{(1 - \beta)/2}}\right)
\end{align*}
almost surely by Theorem \ref{thm:rates}. 

To prove the result in part 2), we 
take the expectation on both sides of \eqref{eq:regret_bound} and using Corollary \ref{cor:delta_k_choice}, we obtain that for all $k \geq 2$, $\delta > 0$, 
\begin{align*}
    &\Eset[R_i(s^k_i, z_i(s^k|G^kP^k))] \\
    &\leq 
    (L^s_J + 2 L^z_J) \Eset[\|s^k - \sbar\|_2] + \Eset[\epsilon_{N,\delta}^k] \\
    &\leq (L^s_J + 2 L^z_J) \sqrt{\frac{N D}{k}} + \Bar{\epsilon}_{N,\delta}(1 - \delta) + 4 L^z_J \smax \delta.
\end{align*}
\end{proof}

The following corollary gives a bound on the \emph{time-averaged}  regret when the step size is chosen appropriately.  

\begin{corollary}
\label{cor: average regret}
    Suppose that Assumptions \ref{a:strat_set}, \ref{a:step_size}, \ref{a:Ls_mu} and \ref{a:smooth-LJz} hold.   
    Let $\{s^k\}_{k \ge 0}$ be a sequence of strategy profiles generated by the gradient dynamics in \eqref{eq:grad_for_dyn_pop}. Fix any $\delta > 0$ and $N \geq 1$. 
    \begin{enumerate}
        \item If the step size sequence is chosen as per Part 1) of Theorem \ref{thm:rates}, then for all $i \in \Ncal$ and for all $T \geq 1$, the time-averaged realized individual regret satisfies
        \begin{align}
            \frac{1}{T} \sum_{k = 1}^T &R_i(s_i^k, \zi{s^k}{G^k P^k}) \nonumber \\
            &\leq \frac{C_4}{T} \sum_{k = 1}^T \|s^k - \sbar\|_2 + \frac{1}{T} \sum_{k = 1}^T \epsilon^k_{N,\delta}. \label{eq:avg_regret_prob}
        \end{align}
        Moreover, for a fixed $N$, the first sum satisfies $(C_4/T) \sum_{k = 1}^T \|s^k - \sbar\|_2 = O\left(1/T^{(1-\beta )/2}\right)$ for any $\beta \in (2 \theta, 1)$, and for a fixed $N$ and $\delta$, the second sum $(1/T)\sum_{k=1}^T \epsilon_{N,\delta}^k$ converges to $(1-\delta) \bar{\epsilon}_{N,\delta} +  4 \delta L_J^z \smax$.
        \item If the step size sequence is chosen as per Corollary \ref{cor:delta_k_choice} with $\alpha=1$, then for all $i \in \Ncal$ and for all $T \geq 2$, the time-averaged expected individual regret satisfies 
        \begin{align*}
            &\Eset\left[\frac{1}{T} \sum_{k = 1}^T R_i(s_i^k,z_i(s^k|G^kP^k))\right] \\
            &\qquad\leq \frac{2 \Bar{J}_1}{T} + 2 C_4 \sqrt{\frac{N D}{T}} + \Bar{\epsilon}_{N,\frac{1}{N}}\left(1 - \frac{1}{N}\right) + \frac{4 L^z_J \smax}{N}.
    \end{align*} 
    \end{enumerate}
\end{corollary}

\begin{proof}
    \begin{enumerate}
        \item By Theorem \ref{thm:regret}, we obtain that for all $T \geq 1$, 
        \begin{align}
            \frac{1}{T} \sum_{k = 1}^T &R_i(s_i^k, \zi{s^k}{G^k P^k}) \nonumber \\
            &\leq \frac{C_4}{T} \sum_{k = 1}^T \|s^k - \sbar\|_2 + \frac{1}{T} \sum_{k = 1}^T \epsilon^k_{N,\delta}, 
        \end{align}
        where $\|s^k-\sbar\|_2 = o\left(1/k^{(1 - \beta)/2}\right)$ for any $\beta \in (2\theta, 1)$.  Thus, there exists $C_1' \geq 0$ and $k_0 \geq 0$ such that for all $k \geq k_0$, 
        \begin{align*}
            \|s^k-\sbar\|_2 &\le \frac{C_1'}{k^{(1 - \beta)/2}}.  
        \end{align*}
        Hence, for any $T \geq k_0$, 
        \begin{align*}
            &\frac{C_4}{T} \sum_{k = 1}^T \|s^k-\sbar\|_2\\
            &\leq \frac{C_4}{T} \sum_{k = 1}^{k_0 - 1} 2 \smax + \frac{C_4}{T} \sum_{k = k_0}^T \frac{C_1'}{k^{(1 - \beta)/2}} \\
            &\leq \frac{2 C_4 \smax (k_0 - 1)}{T} + \frac{C_4 C_1'}{T} \cdot \frac{T^{(\beta + 1)/2}}{(\beta + 1)/2},  
        \end{align*}
        where the last inequality follows from Lemma \ref{lem:series_sum} given in the appendix. 
        Then, the desired result in part 1) follows since we have that
        \begin{align}
            \frac{2 C_4 \smax (k_0 - 1)}{T} + \frac{C_4 C_1'}{T} \cdot \frac{T^{(\beta + 1)/2}}{(\beta + 1)/2} = O(T^{(\beta - 1)/2}) 
        \end{align}
        and the sum $(1/T) \sum_{k=1}^T \epsilon_{N,\delta}^k$ converges to $\Eset[\epsilon_{N,\delta}^k] = (1-\delta) \bar{\epsilon}_{N,\delta} +  4 \delta L_J^z \smax$ by the law of large numbers.

        \item Fix any $N \geq 1$, $i \in \Ncal$ and $\delta > 0$. By Assumption \ref{a:strat_set}, we have the following worst-case bound on the individual regret: 
        \begin{align*}
            &|R_i(s_i^k, \zi{s^k}{G^k P^k}| \\
            &= \Big|J_i(s_i^k, \zi{s^k}{G^k P^k}) - \inf_{s\in\Scal_i} J_i(s_i, \zi{s^k}{G^k P^k})\Big| \\
            &\leq |J_i(s_i^k, \zi{s^k}{G^k P^k})| + \Big|\inf_{s\in\Scal_i} J_i(s_i, \zi{s^k}{G^k P^k})\Big| \\
            &\leq 2 \bar{J}_1 
        \end{align*}
        for all $k \geq 0$. 
        Fix any $T \geq 2$. Then,
        \begin{align*}
            &\hspace{-1cm}\Eset\left[\frac{1}{T} \sum_{k = 1}^T R_i(s_i^k,z_i(s^k|G^kP^k))\right]\\
            &= \frac{1}{T} \sum_{k = 1}^T \Eset\left[R_i(s_i^k,z_i(s^k|G^kP^k))\right]\\
            &\overset{(a)}{\leq} \frac{2 \Bar{J}_1}{T} + \frac{1}{T} \sum_{k = 2}^T \left(L^s_J + 2 L^z_J\right)\sqrt{\frac{N D}{k}} \\
            &\qquad + \Bar{\epsilon}_{N,\delta}(1 - \delta) + 4 L^z_J \smax \delta\\
            &\overset{}{=} \frac{2 \Bar{J}_1}{T} + \left(L^s_J + 2 L^z_J\right)\frac{\sqrt{N D }}{T} \left (\sum_{k = 2}^T \frac{1}{\sqrt{k}} \right) \\
            &\qquad + \Bar{\epsilon}_{N,\delta}(1 - \delta) + 4 L^z_J \smax \delta\\
            &\overset{(b)}{\leq} \frac{2 \Bar{J}_1}{T} + \left(L^s_J + 2 L^z_J\right)\frac{\sqrt{N D }}{T} \left (2\sqrt{T} \right) \\
            &\qquad + \Bar{\epsilon}_{N,\delta}(1 - \delta) + 4 L^z_J \smax \delta \\
            &= \frac{2 \Bar{J}_1}{T} + 2 \left(L^s_J + 2 L^z_J\right)\sqrt{\frac{N D}{T}} \\
            &\qquad + \Bar{\epsilon}_{N,\delta}(1 - \delta) + 4 L^z_J \smax \delta
        \end{align*}
        where (a) follows from Theorem \ref{thm:regret}, and (b) follows from Lemma \ref{lem:series_sum} given in the appendix.
        Hence, the desired result in part 2) follows by choosing $\delta=1/N$.
    \end{enumerate}
\end{proof}

\section{Conclusion} \label{sec:concl}

This paper presents a random network framework to model dynamic populations with time-varying connectivity. 
Using this framework, we study convergence of projected gradient dynamics in repeated time-varying network games and we provide bounds on agents' time-averaged regret.

\section{Appendix}

\begin{lemma} \label{lem:k_alpha}
    Let $\alpha\in(0,1]$. Then, for any $k\ge1$ 
    \begin{align}
        (k+1)^\alpha \le k^\alpha + 1.
    \end{align}
\end{lemma}

\begin{proof}
Let $f(k) := k^\alpha + 1 - (k+1)^\alpha$.
Then, for $k\ge1$,
\begin{align*}
    f'(k) &= \alpha k^{\alpha-1} - \alpha (k+1)^{\alpha-1}\\
    &= \alpha \left( \frac{1}{k^{1-\alpha}} - \frac{1}{(k+1)^{1-\alpha}} \right)\\
    &\ge 0
\end{align*}
where the inequality follows from the fact that $1-\alpha\in[0,1)$ and $1/x^{1-\alpha}$ is a nonincreasing function so $1/k^{1-\alpha} \ge 1/(k+1)^{1-\alpha}$.
Using the fact that $f$ is monotonically increasing for $k\ge 1$, this implies that $f(k) \ge f(1) = 2 - 2^\alpha \ge 0$.
Hence, $(k+1)^\alpha \le k^\alpha + 1$ holds.
\end{proof}

\begin{lemma}
\label{lem:series_sum}
    Let $k_0$ be a positive integer and $T > k_0$. Then, for any $\alpha \in (0,1)$, 
    \begin{align*}
        \sum_{k = k_0}^T \frac{1}{k^\alpha} \leq \frac{T^{1 - \alpha}}{1 - \alpha}. 
    \end{align*}
\end{lemma}

\begin{proof}
    \begin{align*}
        \sum_{k = k_0}^T \frac{1}{k^\alpha} &\leq \int_{k_0}^{T+1} \frac{1}{(x - 1)^\alpha} dx \\
        &= \frac{(x - 1)^{- \alpha + 1}}{- \alpha + 1}\Bigg|_{x = k_0}^{x = T + 1} \\
        &= \frac{T^{-\alpha + 1} - (k_0 - 1)^{-\alpha + 1}}{-\alpha + 1} \\
        &\leq \frac{T^{1 - \alpha}}{1 - \alpha}. 
    \end{align*}
\end{proof}

\bibliographystyle{ieeetr}
\bibliography{references}

\end{document}